\documentclass[10pt,a4paper]{article}
\usepackage{amsmath,amssymb,amsthm,graphicx,braket,multirow,authblk,amsthm,dcolumn,latexsym,subcaption,mathtools,cite,cancel,bm}
\usepackage[normalem]{ulem}
\usepackage[a4paper]{geometry}
\DeclareSymbolFontAlphabet{\amsmathbb}{AMSb}%

\usepackage[unicode=true,
bookmarks=true,bookmarksopen=false,
breaklinks=false,pdfborder={0 0 0},colorlinks=true]
{hyperref}
\usepackage{xcolor}
\usepackage{bbm}

\usepackage{epstopdf}
\definecolor{cblue}{rgb}{0.16, 0.32, 0.75}
\definecolor{cred}{rgb}{0.7, 0.11, 0.11}
\hypersetup{%
	,linkcolor=cred
	,citecolor=cblue
	,urlcolor=black
}
\def\<{\langle}
\def\>{\rangle}
\newtheorem{theorem}{Theorem}

\newtheorem{lemma}{Lemma}

\newcommand{\tr}{\mathop{\mathrm{Tr}}\nolimits}

\renewcommand{\i}{\mathrm{i}}

\DeclareMathAlphabet\mathbfcal{OMS}{cmsy}{b}{n}

\begin{document}	
	
\title{\bf The quantum non-Markovianity for a special class of generalized Weyl channel}

\author[$\hspace{0cm}$]{Wen Xu$^{1,2}$\footnote{wenxu23@foxmail.com}, Mao-Sheng Li$^1$\footnote{li.maosheng.math@gmail.com}, Bo Li$^{3}$\footnote{libobeijing2008@163.com}, Gui-Mei Jiao$^{4}$, Zhu-Jun Zheng$^{1,2}$ \footnote{zhengzj@scut.edu.cn}}

\affil[$1$]{\small School of Mathematics, South China University of Technology, Guangzhou 510640, China}

\affil[$2$] {\small Laboratory of Quantum Science and Engineering, South China University of Technology, Guangzhou 510641, China}

\affil[$3$]{\small School of Computer and Computing Science, Hangzhou City University,
Hangzhou 310015, China}

\affil[$4$]{\small School of Statistics and Mathematics, Lanzhou University, Lanzhou 730000, China}
\date{}
\maketitle
\vspace{-0.5cm}

\noindent {\bf Abstract} {\small }
A quantum channel is usually represented as a sum of Kraus operators. The recent study [\href{https://journals.aps.org/pra/abstract/10.1103/PhysRevA.98.032328}
{Phys. Rev. A \textbf{98}, 032328 (2018)}]
has shown that applying a perturbation to the Kraus operators in qubit Pauli channels, the dynamical maps exhibit interesting properties, such as non-Markovianity, singularity. This has sparked our interest in studying the properties of other quantum channels. In this work, we study a special class of generalized Weyl channel where the Kraus operators are proportional to the Weyl diagonal matrices and the rest are vanishing. We use the Choi matrix of intermediate map to study quantum non-Markovianity. The crossover point of the eigenvalues of Choi matrix is a singularity of the decoherence rates in the canonical form of the master equation. Moreover, we identify the non-Markovianity based on the methods of CP divisibility and distinguishability. We also quantify the non-Markovianity in terms of the Hall-Cresser-Li-Andersson (HCLA) measure and the Breuer-Laine-Piilo (BLP) measure, respectively. In particular, we choose mutually unbiased bases as a pair of orthogonal initial states to quantify the non-Markovianity based on the BLP measure.

\noindent {\bf Keywords}: {\small } Quantum non-Markovianity, Generalized Weyl channels, Intermediate map

\section{Introduction}\label{sec:1}
In the realm of quantum physics, the isolation of a quantum system from its environment is often an idealized scenario. However, such complete isolation is practically unattainable in reality. The interaction between a quantum system and its surrounding environment is inevitable and pervasive, giving rise to a multitude of complex phenomena, such as dissipation and decoherence \cite{Bre-book2007,Wei2000,Riv2014,BLP-V-2016,Veg2017,Li2018,Chr2022}. Environmental noise further exacerbates these effects and thus disrupts the dynamical evolution of quantum processes. The dynamics of open quantum systems are often categorized into two distinct processes: Markovian and non-Markovian, depending on the absence or presence of memory effects. Investigating the characteristic of (non-)Markovian processes is a fundamental and crucial topic that enhances our knowledge of quantum information processing, quantum computation, and quantum communication \cite{Pre2018}.

To study the dynamics of open quantum systems, researchers often employ a mathematical framework based on completely positive and trace-preserving (CPTP) dynamical maps $\mathcal{E}(t)$, which can be used to describe the evolution of the system from an initial state $\rho(0)$ to a final state $\rho(t)$, i.e., $\rho(t)=\mathcal{E}(t)[\rho(0)]$. In many cases, such CPTP linear maps are regarded as quantum channels, which can be expressed as a sum of Kraus operators \cite{Choi-1975}. Quantum channels involve a lot of interesting researches on Markov,  as demonstrated in \cite{Souissi-2021,Wilde-2015,Souissi-2020,Souissi-2024}.

A universally accepted definition for quantum (non-)Markovian dynamics remains elusive \cite{Riv2014,BLP-V-2016,Veg2017,Li2018,Chr2022}. However, there are many approaches to study the quantum non-Markovianity, such as those based on the Lindblad master equations \cite{Linblad1976,GKS1976}, CP divisibility \cite{Wol2008,Riv-2010,Hou2011,Chr2011,Hall-2014}, distinguishability \cite{BLP-2009}, quantum Fisher information flow \cite{Lu-2010,SL2015}, quantum correlation \cite{LFS2012,Jia-Luo2013}, etc. In this paper, we adopt the two methods of CP divisibility and distinguishability to study the quantum non-Markovianity. A dynamical map $\mathcal{E}(t)$ is called CP divisible if and only if it can be decomposed as
\begin{equation}\label{decomp-divisible}
\mathcal{E}(t)=\mathcal{E}(t,s)\;\mathcal{E}(s),
\end{equation}
and the intermediate map $\mathcal{E}(t,s)$ is also CP for any $ 0\leq s\leq t$ \cite{Riv-2010,Chr2011}. The trace distance between a pair of initial states of $\rho_1$ and $\rho_2$ is defined by
\begin{equation}\label{trace-distance}
D(\rho_1,\rho_2)=\frac{1}{2}\tr{|\rho_1-\rho_2|},
\end{equation}
where $|A|=\sqrt{A^\dag A}.$ It satisfies $0\leq D(\rho_1,\rho_2)\leq 1$, and $D(\rho_1,\rho_2)=1$ if and only if $\rho_1$ and $\rho_2$ are orthogonal. The trace distance has a clear physical interpretation in terms of the distinguishability between two quantum states \cite{BLP-2009,Nie-2000}. Moreover, it is non-increasing under any CPTP
map $\mathcal{E}(t)$, i.e.,
$$
D[\mathcal{E}(t)(\rho_1),\mathcal{E}(t)(\rho_2)]\leq D(\rho_1,\rho_2).
$$
This means that a trace preserving quantum operation can never increase the distinguishability of any two quantum states.

The two methodologies yield some criteria and measures for identifying and quantifying the non-Markovianity. Specifically, if the dynamical map $\mathcal{E}(t)$ is CP indivisibility, then it is non-Markovian. In this sense, although the dynamical map $\mathcal{E}(t)$ satisfies the decomposition law \eqref{decomp-divisible}, the associated intermediate map $\mathcal{E}(t,s)$ may not be CP. By
Choi-Jamio{\l}kowski isomorphism \cite{Choi-1975}, we know that the Choi matrix of the intermediate map $\mathcal{E}(t,s)$ is not positive semidefinite, i.e., it has at least one negative eigenvalue, which indicating non-Markovian behavior. Moreover,  Rivas-Hulga-Plenio (RHP) proposed the non-Markovianity measure in terms of not completely positive (NCP) intermediate map \cite{Riv-2010}. From the perspective of Lindblad master equation, Hall-Creser-Li-Anderson (HCLA) proposed the non-Markovianity measure in terms of time-dependent negative decoherence rates \cite{Hall-2014}.
Additionally, the Breuer-Laine-Piilo (BLP) measure \cite{BLP-2009} quantifies the non-Markovianity by tracking the increase in distinguishability between a pair of initial states, which reflects information backflow from the environment to the system.

Recently, the non-Markovianity of qubit Pauli channels has attracted attention in \cite{Shr-2018,Utagi-2020,Abu-2024}. In particular, the non-Markovianity of dephasing and depolarizing channels is investigated by introducing a parameter perturbation to the corresponding Kraus operators \cite{Shr-2018}. This parameter serves as a key indicator of non-Markovian behavior. Perturbation theory has potential physical significance for the study of dynamical maps, such as quantum dynamical semigroups \cite{Hol-2018,Sie-2017}. Moreover, perturbation studies may be extended to investigate how informational characteristics change in various types of quantum channels \cite{Ser-2019,Amo-2020}.

With the rapid development of quantum computing and quantum communication, it has become increasingly urgent to investigate the properties of quantum channels in high-dimensional under the influence of noise \cite{Dutta-2023}. As one of these channels, generalized Weyl channels \cite{Wud2015,Xu-2024} serve as a valuable tool for investigating the non-Markovianity and singularities \cite{Hou2011,Uta2021,Jag2022} of quantum channels.

This paper is organized as follows. In Sec. \ref{sec:2}, we review
the notion of unitary Weyl operators and two matrix representations of quantum channels. In Sec. \ref{sec:3}, we consider a special class of generalized Weyl channel where the Kraus operators are proportional to the Weyl diagonal matrices and the rest are vanishing. Moreover, we obtain the expression of the Choi matrix for intermediate map. We also present a simple case of $d=3$ to discuss the non-Markovianity, one can find that the singularity of the intermediate map occurs at the crossover point among its $d$ eigenvalues. In Sec. \ref{sec:4}, the non-Markovianity is quantified by the HCLA measure, which corresponds to a negative decoherence rate in the canonical master equation. In Sec. \ref{sec:5}, we demonstrate that the singularity of intermediate map is not pathological in the sense of density operator, and the solution of intermediate map is still regular \cite{Chr2010}. In Sec. \ref{sec:6}, we choose mutually unbiased bases as a pair of orthogonal initial states to quantify the non-Markovianity based on the BLP measure. Finally, we conclude in Sec. \ref{sec:7} with discussions on the future research directions.

\section{Weyl operators and two matrix representations of quantum channels}\label{sec:2}
We first recall some definitions and notations that need to be used in this work.

Let $U_{kl}\in\mathbb{C}^{d\times d}$, $k,l=0,1,\ldots,d-1$, be the set of unitary Weyl operators
\begin{equation}\label{Weyl-operator}
U_{kl}=\sum\limits_{m=0}^{d-1}\omega_d^{km}|m\rangle\langle m+l|,~\omega_d=e^{\frac{2\pi \mathrm{i}}{d}},
\end{equation}
where the addition indices in Eq. \eqref{Weyl-operator} are taken over modulo $d$ \cite{Chr2013,Wud2015}. They satisfy
\begin{equation}
U_{kl}U_{rs}=\omega_d^{lr}U_{k+r,l+s},\ U_{kl}^{\dag}=\omega_d^{kl}U_{-k,-l},~r,s=0,1,\ldots,d-1.
\end{equation}

For simplicity, we denote the double subscript of the Weyl operators as a single subscript, i.e., $a=dk+l.$ One has $U_0=\mathbbm{1}_d, \tr(U_{a}U_{b})=d\delta_{ab}\;(
a,b=0,1,\ldots,d^{2}-1)$, where $\mathbbm{1}_d$ is the identity operator in a $d$-dimensional Hilbert space. The Weyl matrices for the case of $d=3$ are given in Appendix A.

A quantum channel $\mathcal{E}$ has the Kraus representation,
\begin{equation}\label{Kraus}
\mathcal{E}(X)=\sum\limits_{a=0}^{d^2-1}K_a
X K_a^{\dag}.
\end{equation}
where $K_a$ are the Kraus operators, which satisfy the completeness condition
\begin{equation}\label{comp-condition}
\sum_{a=0}^{d^2-1}K_a^{\dag}K_a=\mathbbm{1}_d.
\end{equation}

Next, we review two matrix representations of the quantum channel $\mathcal{E}$. Fixing an orthonormal basis $\{|i\rangle,i=0,1,\ldots,d-1\}$ in a Hilbert space, one can define the Choi matrix
\begin{equation}\label{Choi-marix} C=(\mathcal{E}\otimes\mathbbm{1}_d)|\phi^+\rangle\langle\phi^+|=
\sum_{i,j=0}^{d-1}\mathcal{E}(|i\rangle\langle j|)\otimes|i\rangle\langle j|,
\end{equation}
where $|\phi^+\rangle=\sum_{i=0}^{d-1}|ii\rangle$ is the unnormalized maximally entangled state, and the matrix elements of Hermitian Choi matrix $C \in \mathbb{C}^{d^2\times d^2}$ are given by
$C_{ij,kl}=\langle i\otimes j| C|k\otimes l\rangle=\langle i|\mathcal{E}(|j\rangle\langle l|)|k\rangle.$ By
Choi-Jamio{\l}kowski isomorphism, the quantum channel $\mathcal{E}$ is completely positive (CP) if and only if the Choi matrix $C$ is positive semidefinite \cite{Choi-1975}.

Any matrix $X\in\mathbb{C}^{d\times d}$ can be mapped to a vector $|X\rangle\rangle\in\mathbb{C}^{d}\otimes\mathbb{C}^{d}$ as follows:
\begin{equation}\label{vec}
|X\rangle\rangle=\sum_{i,j=0}^{d-1} X_{ij}|ij\rangle,
\end{equation}
where $X_{ij}=\langle i|X|j\rangle.$ The form of Eq. \eqref{vec} is called the vectorization of $X$ \cite{Zyc-2004,Jag-2018,Tar-2021}. One can define a superoperator $\widehat{\mathcal{E}}\in\mathbb{C}^{d^2\times d^2}$ via
\begin{equation}\label{superoperator}
\widehat{\mathcal{E}}|X\rangle\rangle=|\mathcal{E}(X)\rangle\rangle,
\end{equation}
and the corresponding matrix elements of $\widehat{\mathcal{E}}$ are given by $\widehat{\mathcal{E}}_{ij,kl}=\langle ij|\widehat{\mathcal{E}}|kl\rangle$ \cite{Tar-2021}.

The Choi matrix and superoperator representations of $\mathcal{E}$ are related by the reshuffling operation \cite{Zyc-2004,Tar-2021}
\begin{equation}\label{reshuffling}
 C=\widehat{\mathcal{E}}^{\mathcal{R}},~~
 C_{ij,kl}=\widehat{\mathcal{E}}_{ik,jl}.
\end{equation}
By the Kraus representation Eq. \eqref{Kraus} of the quantum channel $\mathcal{E}$, the Choi matrix Eq. \eqref{Choi-marix} can be rewritten as
\begin{equation}\label{Choi-representation}
 C =\sum_{a=0}^{d^2-1}|K_a\rangle\rangle\langle\langle K_a|,
\end{equation}
and the superoperator Eq. \eqref{superoperator} is given by
\begin{equation}\label{superoperator-representation}
\widehat{\mathcal{E}}=\sum_{a=0}^{d^2-1}K_a\otimes \overline{K}_a,
\end{equation}
where $\overline{K}_a$ is the conjugate of $K_a$. The proof of Eqs. \eqref{Choi-representation} and \eqref{superoperator-representation} are given in Appendix B.

\section{A special class of generalized Weyl channel and the Choi matrix for the intermediate map}
\label{sec:3}   

\subsection{A special class of generalized Weyl channel}
In this subsection, we focus on a special class of generalized Weyl channel. Unlike Ref. \cite{Abu-2024}, we introduce a new technical approach (i.e., the use of two matrix representations of quantum channels) to derive the expression of Choi matrix corresponding to the intermediate map. This allows us to analyze the non-Markovianity and singularity of generalized Weyl channels.

Generalized Weyl channel was introduced in Ref. \cite{Wud2015,Xu-2024}, which can be represented by a sum of Kraus operators in Eq. \eqref{Kraus}. It is natural that  these Kraus operators can be given by
\begin{equation}\label{k-parameterization}
K_0=\sqrt{1-\frac{d^2-1}{d^2}\kappa}\;\mathbbm{1}_d, ~~ K_a=\sqrt{\frac{\kappa}{d^2}}\;U_a,\; a=1,\ldots,d^2-1,
\end{equation}
where $\{1-\frac{d^2-1}{d^2}\kappa, \frac{\kappa}{d^2}, \ldots, \frac{\kappa}{d^2}\}$ is the set of probability distribution and $\kappa$ is the mixing parameter of quantum channel.

Motivated by the method of Ref. \cite{Shr-2018,Abu-2024}, one can introduce a time-like parameter $p$ to generalize the form of  Eq. \eqref{k-parameterization}, which can be written as 
\begin{equation}\label{p-parameterization}
K_0= \sqrt{\left[1+\Lambda_0(p)\right]\left(1-\frac{d^2-1}{d^2}p\right)}\;\mathbbm{1}_d, ~~ K_a=\sqrt{\left[1+\Lambda_a(p)\right]\;\frac{p}{d^2}}\;U_a,\;a=1,\ldots,d^2-1.
\end{equation}
Here $\Lambda_a(p)$ are real functions for all $a=0,1,\ldots,d^2-1$, and $p$ is the time-like parameter, which increases monotonically from $0$ to $1$. In fact, time-like refers to the parameter $p$ acting similarly to the time-dependent probability distribution function $p(t)$, which increases monotonically with time $t$, but we do not care about its detailed functional dependence \cite{Shr-2018}. In particular, when $\Lambda_a(p)=0$ for all $a=0,1,\ldots,d^2-1$ and $p$ is replaced by $\kappa$, then Eq. \eqref{p-parameterization} reduces to Eq. \eqref{k-parameterization}.

In this paper, we use the Kraus operators with the time-like parameter $p$ to study the conditions under which the generalized Weyl channel is non-Markovian. According to the above $d^2$ Kraus operators in Eq. \eqref{p-parameterization}, one can calculate the Choi matrix corresponding to the intermediate map. However, its expression is very complicated. For simplicity, we consider a special class of generalized Weyl channel. That is, we retain only the $d$ Kraus operators in Eq. \eqref{Kraus}, which have the following form
\begin{equation}\label{Weyl-Kraus}
K_{0}=\sqrt{[1+\Lambda_0(p)]\left(1-\frac{d-1}{d}p\right)}\;\mathbbm{1}_d,~~
K_{di}=\sqrt{[1+\Lambda_i(p)]\;\frac{p}{d}}\;U_{di},\;i=1,\ldots,d-1,
\end{equation}
and the other Kraus operators are vanishing. Here $U_{di}$ are Weyl diagonal matrices for all $i=1,\ldots,d-1$. Without confusion, the subscripts $di$ of $U_{di}$ and $K_{di}$ refer to the scalar products of $d$ and $i$. For the case of $d=3$, the Weyl diagonal matrices $U_3$ and $U_6$ are given in Appendix A. By the completeness condition Eq. \eqref{comp-condition}, one has
\begin{equation}
\left(1-\frac{d-1}{d}p\right)\Lambda_0(p)+\frac{p}{d}\sum_{i=1}^{d-1}\Lambda_i(p)=0.
\end{equation}
We choose $\Lambda_0(p)=-\frac{d-1}{d}\alpha p,$
$\Lambda_i(p)=\alpha\left(1-\frac{d-1}{d}p\right)$ for all $i=1,\ldots,d-1$, where $\alpha$ is a real parameter, and $p\in[0,1]$. Then we have
\begin{equation}\label{parameterization-kraus}
K_{0}=\sqrt{1-\frac{d-1}{d}\kappa(p)}\;\mathbbm{1}_d,~
K_{di}=\sqrt{\frac{\kappa(p)}{d}}\;U_{di},~i=1,\ldots,d-1,
\end{equation}
where $\kappa(p)=p\left[1+\alpha\left(1-\frac{d-1}{d}p\right)\right].$ In this case, the generalized Weyl channel $\mathcal{E}(p)$ is given by
\begin{align}\label{special-Weyl}
\mathcal{E}(p)(\rho)=&\left[1-\frac{d-1}{d}\kappa(p)\right]\rho+
\frac{\kappa(p)}{d}\sum_{i=1}^{d-1}U_{di}\rho U_{di}^\dag \nonumber\\
=&\frac{1}{d-1}\sum_{i=1}^{d-1}\left\{
\left[1-\frac{d-1}{d}\kappa(p)\right]\rho+
\frac{d-1}{d}\kappa(p)U_{di}\rho U_{di}^\dag\right\}.
\end{align}
It is essentially an average convex combination of $(d-1)$ generalized Weyl dephasing channels. Here the generalized Weyl dephasing channels are of the form $\mathcal{E}_i(\rho)=\left[1-\frac{d-1}{d}\kappa(p)\right]\rho+
\frac{d-1}{d}\kappa(p)U_{di}\rho U_{di}^\dag$ for all $i=1,\ldots,d-1$ (for more details on generalized Weyl dephasing channels, see Ref. \cite{Xu-2024}).

To ensure the generalized Weyl map $\mathcal{E}(p)$ is CP, we choose the parameter $\alpha\in[0,1]$. The introduction of the parameter $\alpha$ has a small perturbation effect on the generalized Weyl channel, which leads to the study of non-Markovian dynamics. We also regard it as a non-Markovian parameter.

\subsection{The Choi matrix for the intermediate map}\label{subsec:3.2}
Now we consider the dynamics on the time-like parameter $p$ of generalized Weyl channel $\mathcal{E}(p)$, which evolving the initial state $\rho(0)$ to the final state $\rho(p)$, i.e., $\rho(p)=\mathcal{E}(p)(\rho(0))$. For $p\in[p_*,p^*]$  with $0\leq p_*< p^*\leq 1,$ similar to Eq. \eqref{decomp-divisible}, the CP map $\mathcal{E}(p^*)$ is CP-divisible if it
can be decomposed as
\begin{equation}\label{divisible}
\mathcal{E}(p^*)=\mathcal{E}(p^*,p_*)\;\mathcal{E}(p_*),
\end{equation}
and the intermediate map $\mathcal{E}(p^*,p_*)$ is also CP for all $0\leq p_*< p^* \leq 1.$

From Eq. \eqref{divisible}, the intermediate map $\mathcal{E}(p^*,p_*)$ can be written as
\begin{equation}\label{invertible}
\mathcal{E}(p^*,p_*)=\mathcal{E}(p^*)\;\mathcal{E}^{-1}(p_*)
\end{equation}
if the map $\mathcal{E}^{-1}(p_*)$ is invertible. The Choi matrix for the intermediate map $\mathcal{E}(p^*,p_*)$ is given by
\begin{equation}
\chi(\alpha,p^*,p_*)=[\mathcal{E}(p^*,p_*)\otimes \mathbbm{1}_d]|\phi^+\rangle\langle\phi^+|,
\end{equation}
where $|\phi^+\rangle=\sum_{i=0}^{d-1}|ii\rangle$ is the unnormalized maximally entangled state. 

By Choi-Jamio{\l}kowski isomorphism \cite{Choi-1975}, the intermediate map $\mathcal{E}(p^*,p_*)$ is CP if and only if the corresponding Choi matrix $\chi(\alpha,p^*,p_*)$ is positive semidefinite. Therefore, CP divisibility of dynamical map is closely related to the Choi matrix corresponding to intermediate map. We have the the following non-Markovian criterion.

\textbf{Criterion 1:} If the Choi matrix $\chi(\alpha,p^*,p_*)$ has negative eigenvalues, then the dynamical map $\mathcal{E}(p^*)$ is non-Markovian. Otherwise, it is Markovian.

To investigate the non-Markovianity of generalized Weyl channels, we first show the expression of Choi matrix for the intermediate map $\mathcal{E}(p^*,p_*)$.

\begin{theorem}\label{Thm1}
For a special class of generalized Weyl channel $\mathcal{E}(p)$ and the intermediate map is given by Eq. \eqref{special-Weyl} and Eq. \eqref{divisible}, respectively. Then one can obtain the Choi matrix 
\begin{equation}\label{Choi-inter}
\chi(\alpha,p^*,p_*)=\left[\widehat{\mathcal{E}}(p^*)
\widehat{\mathcal{E}}^{-1}(p_*)\right]^{\mathcal{R}}
=
\begin{pmatrix}
E_{00} & \widetilde{E}_{01} & \cdots & \widetilde{E}_{0,d-1} \\
\widetilde{E}_{10} & E_{11} & \cdots & \widetilde{E}_{1,d-1} \\
\vdots & \vdots & \ddots & \vdots \\
\widetilde{E}_{d-1,0} & \widetilde{E}_{d-1,1} & \cdots & E_{d-1,d-1}
\end{pmatrix}.
\end{equation}
Here $\widehat{\mathcal{E}}(p)$ is the superoperator matrix given by Eq. \eqref{superoperator-representation}, $\mathcal{R}$ is the reshuffling operation,
$\widetilde{E}_{ij}=\frac{G(p^*)}{G(p_*)}E_{ij}
=\frac{1-\kappa(p^*)}{1-\kappa(p_*)}E_{ij}\; (i\neq j, i,j=0,1,\ldots,d-1)$, and each $E_{ij} \;(i,j=0,1,\ldots,d-1)$ is a $d\times d$ matrix whose the entry in row $(i+1)$ and column $(j+1)$ is $1$ and all the others are $0$, where the function $G(p)$ is given by
\begin{equation}\label{G-function}
G(p)=1-\kappa(p)=\frac{d-1}{d}\alpha p^2-(1+\alpha)p+1.
\end{equation}
\end{theorem}
The proof of Theorem \ref{Thm1} is given in Appendix C.

{\sf Remark 1:} The Choi matrix $\chi(\alpha,p^*,p_*)$ in Eq. \eqref{Choi-inter} is composed of multiple block matrices and has only $d^2$ nonzero matrix elements. The simple cases of $d=2,3$ are given in Appendix C. Theorem \ref{Thm1} can be regarded as a generalization of Ref. \cite{Shr-2018}. Based on the special structure of the Choi matrix, we derive its eigenvalues to analyze non-Markovianity. To provide an intuitive understanding, we visualize the non-Markovianity through graphical representations (see subsection 3.3 for details). Additionally, this structure allows us to analyze the impact of noise on quantum channels (via a perturbation in parameter $\alpha$) and quantify the degree of interference by using the non-Markovianity measures such as HCLA and BLP (see sections \ref{sec:4} and \ref{sec:6}), which offering significant potential to advance research on noise resistance in quantum systems.

{\sf Remark 2:} Compared to the method used in Ref. \cite{Abu-2024}, one of our advantages is that we employ a new technical approach (i.e., the use of two matrix representations of quantum channel) to derive the expression of Choi matrix. Then we can get the analytic forms of its eigenvalues, avoiding the need for numerical calculations.

The function $G(p)$ has two real roots
\begin{equation}\label{two-roots}
\alpha_{\pm}=\frac{d}{d-1}\times\frac{1+\alpha\pm\sqrt{(1+\alpha)^2-4\alpha(d-1)/d}}{2\alpha}.
\end{equation}
Note that $\lim_{\alpha\rightarrow0^+}\alpha_-=1$ and $\alpha_-$ decreases monotonically for $\alpha\in(0,1]$, one has $0<\alpha_-\leq1.$ On the other hand, since $\alpha_{+}\alpha_{-}=\frac{d}{(d-1)\alpha}>1$ and $\alpha_{\pm}>0,$ we have $\alpha_+>1$.

The eigenvalues of the Choi matrix in Eq. \eqref{Choi-inter} are
\begin{equation}\label{eigenvalue}
\lambda_0=1+(d-1)\frac{(\alpha_{-}-p^*)(\alpha_{+}-p^*)}
{(\alpha_{-}-p_*)(\alpha_{+}-p_*)},~~
\lambda_i=1-\frac{(\alpha_{-}-p^*)(\alpha_{+}-p^*)}
{(\alpha_{-}-p_*)(\alpha_{+}-p_*)}, ~i=1,\ldots,d-1,
\end{equation}
and the other $d(d-1)$ eigenvalues are zero. This result is a generalization of Ref. \cite{Shr-2018}. When $d=2$, the eigenvalues of Eq. \eqref{eigenvalue} are coincided with Ref. \cite{Shr-2018}.

To simplify notation, we also denote the intermediate map as $\mathcal{E}^{\text{int}}(p):=\mathcal{E}(p^*, p_*), \;p\in[p_*,p^*].$ By the eigenvalues relation of Eq. \eqref{eigenvalue}, one can get the  Kraus operators for the intermediate map
\begin{equation}
K_j^{\text{int}}=\sqrt{\epsilon_j\lambda_j/d}~U_{dj},\; j=0,1,\ldots,d-1,
\end{equation}
and the other $d(d-1)$ Kraus operators of intermediate map are vanishing, where
$$
\epsilon_j=
\begin{cases}
1,& \; \text{if}~ \lambda_j>0,\\
-1,& \; \text{if}~ \lambda_j<0.
\end{cases}
$$
Hence the intermediate map is given by
\begin{equation}\label{int-map}
\mathcal{E}^{\text{int}}(p)(\rho)=\sum_{j=0}^{d-1}\epsilon_j K_j^{\text{int}}\rho
K_j^{\text{int}^\dag}=\frac{1}{d}
\sum_{j=0}^{d-1}\lambda_jU_{dj}\rho U^\dag_{dj}=
\rho+\frac{1}{d}\sum_{j=1}^{d-1}\lambda_j\left(U_{dj}\rho U^\dag_{dj}-\rho\right),
\end{equation}
and the completeness relation is $\sum_{j=0}^{d-1}\epsilon_j K_j^{\text{int}^\dag}K_j^{\text{int}}=\mathbbm{1}_d,$ where the last equality holds as $\sum_{j=0}^{d-1}\lambda_j=d$.

In the case of no perturbation, i.e., $\alpha=0,\; p_*=0$, the eigenvalues of Choi matrix in Eq. \eqref{eigenvalue} are reduced to
\begin{equation}\label{no-perturb-eig}
\lambda_0=d-(d-1)p^*,~~
\lambda_i=p^*, ~i=1,\ldots,d-1.
\end{equation}
Since these eigenvalues are nonnegative for the whole range $0\leq p^*\leq1$, then the intermediate map $\mathcal{E}(p^*,p_*)$ is CP. Hence, the generalized Weyl map $\mathcal{E}(p)$ is Markovian for $0\leq p^*\leq1$. For the case $d=3$, the eigenvalues of Eq. \eqref{no-perturb-eig} are depicted in Fig. \ref{fig1}. One finds that the crossover point of the eigenvalues lies in $p^*=1.$

On the other hand, the main work of this paper is to introduce a perturbation $\alpha\in(0,1]$ into the dynamics, which is also generally considered to be a non-Markovian parameter. In the following, we will discuss the non-Markovianity of this perturbation on the intermediate map.

\subsection{A simple case of $d=3$}
In this subsection, we investigate the non-Markovianity for a special class of generalized Weyl channel. In subsection \ref{subsec:3.2}, we have obtained the Choi matrix for the intermediate map $\mathcal{E}(p^*,p_*)$, i.e., Eq. \eqref{Choi-inter}. Next, we utilize the eigenvalues of the Choi matrix $\chi(\alpha,p^*,p_*)$ to discuss the non-Markovianity of generalized Weyl channel.

Note that the eigenvalues $\lambda_j(\alpha,p^*,p_*), j=0,1,\ldots,d-1,$ in Eq. \eqref{eigenvalue} are the functions of the variables $\alpha, p^*,$ and $p_*$; here we adopt the similar method of Ref. \cite{Shr-2018} to study the non-Markovianity of generalized Weyl channel. We fix the dimension $d$ of the system, such as the simple case of $d=3$ (the other high-dimensional cases are discussed similarly). In addition, we also fix the parameters $\alpha$ and $p_*$. Under the above conditions, these eigenvalues reduce to the functions $\lambda_j(p^*)$ that are only related to $p^*$.

Taking the simple case of $d=3$ as an example, we discuss the non-Markovianity for the special class of generalized Weyl channel. Here we fix the parameter $\alpha=0.5$ and choose two different parameters $p_*$ to investigate how these eigenvalues $\lambda_j(p^*)$ change over the range $p^*\in[p_*,1]$. For $d=3$ and $\alpha=0.5$, by Eq. \eqref{two-roots}, one has $\alpha_{-}\approx0.81, \alpha_{+}\approx3.69.$ We discuss the following two cases: $p_*<\alpha_{-}$ and $p_*>\alpha_{-}.$

\begin{enumerate}
\item [\rm (i)]
If $p_*=0.3$ and $p^*\in[0.3,1]$, the eigenvalues are reduced to
$$
\lambda_0=1+2\times\frac{(0.81-p^*)(3.69-p^*)}
{0.51\times3.39},~~
\lambda_1=\lambda_2=1-\frac{(0.81-p^*)(3.69-p^*)}
{0.51\times3.39}.
$$
From Fig. \ref{fig1}, one observes that at the point $p^*=p_*=0.3,$ one has $\lambda_0=3, \lambda_1=\lambda_2=0;$ at the point $p^*=\alpha_{-}\approx0.81$, these eigenvalues are the same, i.e., $\lambda_0=\lambda_1=\lambda_2=1$. Since all the eigenvalues $\lambda_j>0\;(j=0,1,2)$ over the range $p^*\in[0.3,1]$, then the generalized Weyl channel defined by Eq. \eqref{special-Weyl} is Markovian.

\item [\rm (ii)]
If $p_*=0.85$ and $p^*\in[0.85,1]$, the eigenvalues are reduced to
$$
\lambda_0=1+2\times\frac{(0.81-p^*)(3.69-p^*)}
{(-0.04)\times2.84},~~
\lambda_1=\lambda_2=1-\frac{(0.81-p^*)(3.69-p^*)}
{(-0.04)\times2.84}.
$$
From Fig. \ref{fig2}, one observes that the eigenvalues $\lambda_i\;(i=1,2)$ are negative in the entire range of $p^*\in(0.85,1].$ It suggests that
the intermediate map is NCP, i.e., the generalized Weyl channel defined by Eq. \eqref{special-Weyl} is non-Markovian.
\end{enumerate}

\begin{figure}[htbp]
\centering
\setlength{\abovecaptionskip}{0.05cm}
{\includegraphics[width=7cm,height=6cm]{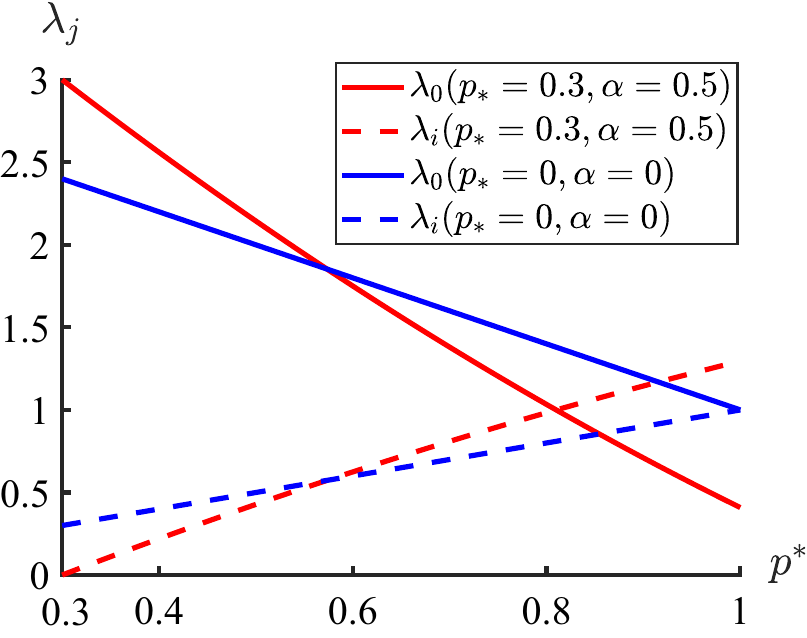}}
\caption{\footnotesize{The eigenvalues of Choi matrix Eq. \eqref{Choi-inter} of intermediate map for the case of $d=3$. The eigenvalues $\lambda_0$ (solid, blue line) and $\lambda_i, i=1,2,$ (dashed, blue line) for $p_*=0, \alpha=0$. The eigenvalues $\lambda_0$ (solid, red line) and $\lambda_i, i=1,2,$ (dashed, red line) for $p_*=0.3, \alpha=0.5$. The crossover point of the eigenvalues $\lambda_0$ and $\lambda_i$ in these two case is $1$ and $0.81$, respectively. For the two cases, the eigenvalues of intermediate map are nonnegative, so the
generalized Weyl channel $\mathcal{E}(p)$ is Markovian for the whole range $p^*\in[0.3,1]$.}}\label{fig1}
\end{figure}

\begin{figure}[htbp]
\centering
\setlength{\abovecaptionskip}{0.05cm}
{\includegraphics[width=7cm,height=6cm]{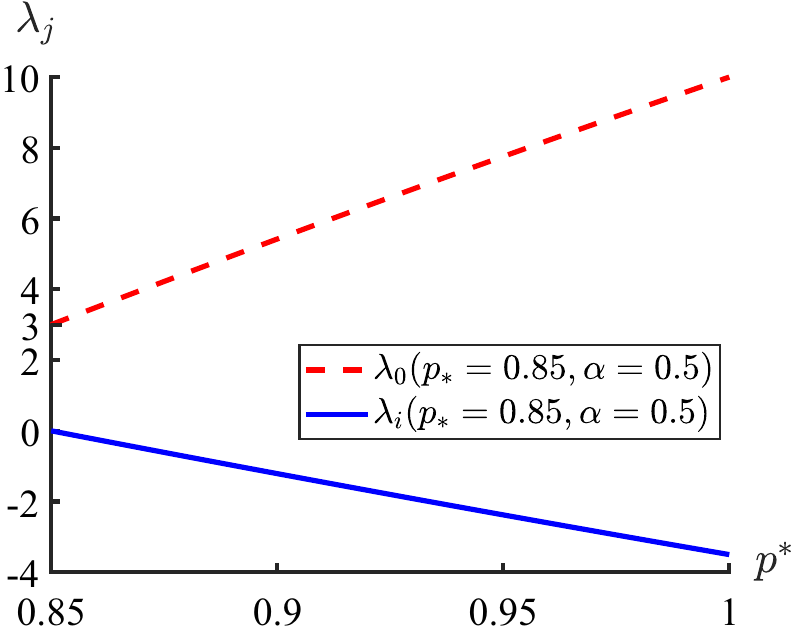}}
\caption{\footnotesize{The eigenvalues of Choi matrix Eq. \eqref{Choi-inter} of intermediate map for the case of $d=3$. The eigenvalues $\lambda_0$ (dashed, red line) and $\lambda_i, i=1,2,$ (solid, blue line) for $p_*=0.85, \alpha=0.5$. For the parameter $p^*\in(0.85,1],$ one finds $\lambda_i<0, i=1,2.$ Thus, the whole range of $p^*\in(0.85,1]$ corresponds to a NCP intermediate map, which indicates the non-Markovianity of generalized Weyl channel.}}\label{fig2}
\end{figure}

{\sf Remark 3:} This simple case shows that the non-Markovianity of generalized Weyl channel depends on the choice of the parameters $\alpha$ and $p_*$. Moreover, the above discussions can be naturally extended to arbitrary high-dimensional cases in a similar way.

Based on the above discussions, we have the following two observations:
\begin{enumerate}
\item [\rm (1)]
In case (i), the crossover point $p_*=\alpha_{-}$ of the eigenvalues represents a singular point of the intermediate map, since the eigenvalues $\lambda_j\; (j=0,1,\ldots,d-1)$ diverge for any $p^*\in(p_*,1].$ The other point $\alpha_{+}>1$ is outside the domain of $p^*$, so it does not need to consider for this study.

\item [\rm (2)]
In case (ii), one finds that if $p^*\rightarrow p_*$, i.e., $p^*\rightarrow 0.85$, then the eigenvalue $\lambda_i\rightarrow0^- (i=1,2)$. This suggests that the instantaneous intermediate map is NCP. The NCP character of this intermediate map is essentially related to the non-Markovianity based on Rivas-Huelga-Plenio (RHP) measure  \cite{Riv-2010}. One can easily find that the trace of Choi matrix $\tr[\chi(\alpha,p^*,p_*)]=d$ from the eigenvalues relation Eq. \eqref{eigenvalue}. We claim that the trace norm of the normalized Choi matrix, $\left\|\frac{1}{d}\chi(\alpha,p^*,p_*)\right\|_{1}=
\left\|[\mathcal{E}(p^*,p_*)\otimes\mathbbm{1}_d]
(\frac{1}{d}|\phi^+\rangle\langle\phi^+|)\right\|_{1}=1$, if and only if the intermediate map $\mathcal{E}(p^*,p_*)$ is CP \cite{Riv-2010}. In fact, note that the intermediate map $\mathcal{E}(p^*,p_*)$ is trace preserving and the Choi matrix $\chi(\alpha,p^*,p_*)$ is Hermitian, $\left\|\frac{1}{d}\chi(\alpha,p^*,p_*)\right\|_{1}=
\frac{1}{d}\sum_{j=0}^{d-1}|\lambda_j|=1$ if and only if the eigenvalues $\lambda_j\geq0$ (i.e., the Choi matrix $\chi(\alpha,p^*,p_*)$ is positive semidefinite), which is equivalent to the intermediate map $\mathcal{E}(p^*,p_*)$ is CP by the Choi-Jamio{\l}kowski isomorphism \cite{Choi-1975}. Hence, the negative eigenvalues for the Choi matrix imply that its normalized trace norm $\left\|\frac{1}{d}\chi(\alpha,p^*,p_*)\right\|_{1}>1$, which provides a witness of the NCP character of $\mathcal{E}(p^*,p_*)$. We can define
$$
g(p)=\lim_{\epsilon\rightarrow0^+}
\frac{\left\|[\mathcal{E}(p^*,p_*)\otimes\mathbbm{1}_d]
(\frac{1}{d}|\phi^+\rangle\langle\phi^+|)\right\|_{1}-1}{\epsilon},
$$
where $\epsilon=p^*-p_*$. It is clear that $g(p)\geq0$, and $g(p)=0$ if and only if the generalized Weyl channel is Markovian. The integral $\mathcal{I}=\int_{0}^{1}g(p)\mathrm{d}p$ provides a quantification  measure of non-Markovianity, which is the Rivas-Huelga-Plenio (RHP) measure \cite{Riv-2010}.
\end{enumerate}

On the other hand, compared to the RHP measure, Hall-Cresser-Li-Andersson (HCLA) proposed the other non-Markovianity measure based on negative decoherence rates \cite{Hall-2014}, which was a possibly computationally easier method. In the next section, we use the decoherence rates in the canonical form of the master equation to discuss the non-Markovianity of the generalized Weyl channel.

\section{Negative decoherence rates in the canonical form of the master equation}
\label{sec:4}
The time-like (i.e., parameter $p$) evolution of the system can be written as the canonical form of the master equation \cite{Hall-2014}
\begin{equation}
\frac{\mathrm{d}\rho(p)}{\mathrm{d}p}=-\mathrm{i}[H(p),\rho(p)]+
\sum_i\gamma_i(p)\left(L_i(p)\rho(p)L_i^\dag(p)-
\frac{1}{2}\{L_i^\dag(p) L_i(p),\rho(p)\}\right),
\end{equation}
where $H(p)$ is the Hamiltonian, $\gamma_i(p)$ are the decoherence rates, and $L_i(p)$ are the traceless orthonormal operators [i.e., $L_i(p)$ satisfy $\tr[L_i(p)]=0, \tr[L_i(p)L_j^\dag(p)]=\delta_{ij}$]. 

Based on the above canonical form of the master equation, we have the following non-Markovian criterion. 

\textbf{Criterion 2:} If there exists some decoherence rate $\gamma_i(p)<0$, then the generalized Weyl map $\mathcal{E}(p)$ in Eq. \eqref{special-Weyl} is non-Markovian. Otherwise, it is Markovian.

The evolved state $\rho(p)=\mathcal{E}^{\text{int}}(p)[\rho(0)]$ for the intermediate map $\mathcal{E}^{\text{int}}(p)$ in Eq. \eqref{int-map}, which satisfies the following master equation in the canonical form
\begin{equation}\label{master-equation}
\frac{\mathrm{d}\rho(p)}{\mathrm{d}p}=
\sum_{i=1}^{d-1}\gamma_i(p)\left[U_{di}\rho(p)U_{di}^\dag-
\rho(p)\right],
\end{equation}
where $\gamma_i(p)=\gamma(p)$ for $i=1,\ldots,d-1.$ Note that the function $G(p)$ is given by Eq. \eqref{G-function}; we can obtain the expression of the decoherence rates $\gamma(p)$ as follows:
\begin{equation}\label{decoherence-rate}
\gamma(p)=-\frac{1}{d}\times\frac{\dot{G}(p)}{G(p)}=
\frac{1+\alpha-\frac{2(d-1)}{d}\alpha p}{(d-1)\alpha p^2-d(1+\alpha)p+d}=
\frac{2}{d}\times \frac{\frac{1}{2}(\alpha_{+}+\alpha_{-})-p}
{(p-\alpha_{-})(p-\alpha_{+})},
\end{equation}
where $\dot{G}(p):=\frac{\mathrm{d}G(p)}{\mathrm{d}p},$ $\alpha_{\pm}$ are the two real roots in Eq. \eqref{two-roots} and $\alpha_{+}+\alpha_{-}=\frac{d}{d-1}\times\frac{1+\alpha}{\alpha},  \alpha\in(0,1]$. In particular, for $\alpha=0$, $G(p)=1-p$ and $\gamma(p)=\frac{1}{d(1-p)}$. Since $\frac{\alpha_{+}+\alpha_{-}}{2}=\frac{d}{2(d-1)}(1+\frac{1}{\alpha})>1\geq p$ for the non-Markovian parameter $\alpha\in(0,1]$; if $p<\alpha_-$, the decoherence rate $\gamma(p)\geq0,$ then the generalized Weyl map $\mathcal{E}(p)$ is Markovian. Otherwise, it becomes non-Markovian if $\alpha_-<p\leq1.$ Here $p=\alpha_-$ is a singular point of the intermediate map. The decoherence rate $\gamma(p)$ in Eq. \eqref{decoherence-rate} is a generalization of Ref. \cite{Shr-2018}.

With an appropriate choice of the parameter $\alpha$, we can demonstrate that the decoherence rate $\gamma(p)<0$ under certain conditions. Below, we present a simple example, and the other cases of high-dimensional can be similarly verified.

{\it Example 1:} For $d=3$ and $\alpha=0$ (the case of no perturbation), the decoherence rate $\gamma(p)=\frac{1}{3(1-p)}>0$, which implies the generalized Weyl channel is Markovian. If we take the parameter $\alpha=0.8,$ one has $\alpha_-\approx0.7, \alpha_+\approx2.67$, the decoherence rate $\gamma(p)=\frac{\frac{9}{8}-\frac{2}{3}p}{(p-0.7)(p-2.67)}.$ One can obtain that the decoherence rate $\gamma(p)<0$ in the range of $p\in(0.7,1]$, which indicates the generalized Weyl channel is non-Markovian. In Fig. \ref{fig3}, we plot the decoherence rate $\gamma(p)$ with respect to the parameter $\alpha=0$ and $\alpha=0.8,$ respectively.

\begin{figure}[htbp]
\centering
\setlength{\abovecaptionskip}{0.05cm}
{\includegraphics[width=7cm,height=6cm]{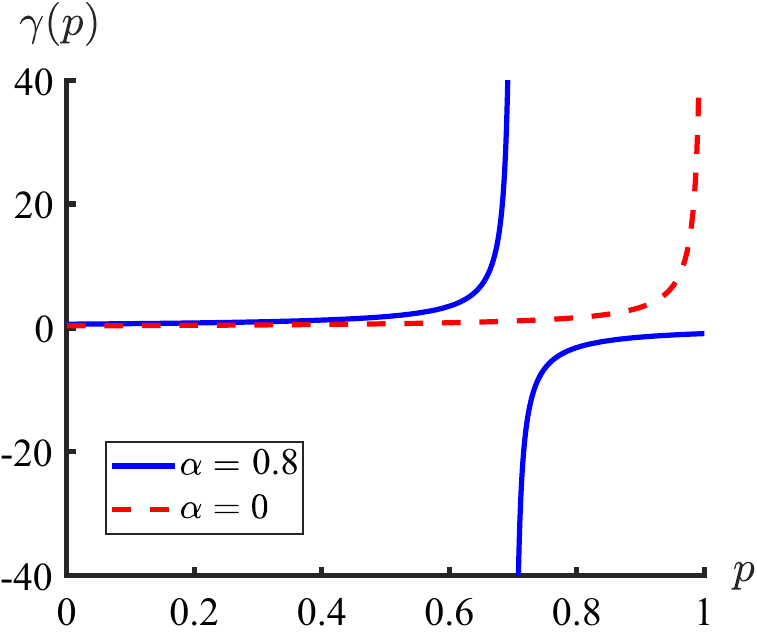}}
\caption{\footnotesize{A plot of the decoherence rates $\gamma_i(p)=\gamma(p)\; (i=1,\ldots,8)$ for the case of $d=3$. The decoherence rate $\gamma(p)$ is a function of $p$ for $\alpha=0.8$ (solid, blue curve) and $\alpha=0$ (dashed, red curve). The generalized Weyl channel is Markovian in the range $p\in[0,1]$ for $\alpha=0$. While for $\alpha=0.8$, the decoherence rate $\gamma(p)<0$ at the range of $p\in(0.7,1]$, which indicates the generalized Weyl channel is non-Markovian.}}\label{fig3}
\end{figure}

Let us recall that the non-Markovianity measure of HCLA \cite{Hall-2014}. Note that the canonical decoherence rates $\gamma_i(p)$ in  Eq. \eqref{decoherence-rate} may exist the negative ones; we can define
$$
f_i(p)=\max\left\{0,-\gamma_i(p)\right\}.
$$
It is clear that $f_i(p)\geq0$ for all $i=1,\ldots,d-1$ and $f_i(p)=0$  if and only if the generalized Weyl map is Markovian. It is non-Markovian if $f_i(p)>0,$ one can use the function $f_i(p)$ to quantify the non-Markovianity as
$$
\mathcal{N}_{\text{HCLA}}=-\int_{\alpha_-}^1\gamma(p)\mathrm{d}p,
$$
where $\gamma(p)$ is defined by Eq. \eqref{decoherence-rate}. However, the function $\gamma(p)$ diverges at the singular point $p=\alpha_-$. To avoid this problem, following an idea proposed in \cite{Riv2014,Shr-2018}, we replace $-\gamma(p)$ by its normalized version
\begin{equation}
\gamma^{\prime}(p):=\frac{-\gamma(p)}{1-\gamma(p)}=
\frac{\frac{2(d-1)}{d}\alpha p-(1+\alpha)}{h(\alpha,p)},
\end{equation}
where $h(\alpha,p)=(d-1)\alpha p^2-\left[\left(d-\frac{2(d-1)}{d}\right)\alpha+d\right]p+d-1-\alpha$. Hence, we get a normalized  HCLA measure
\begin{align}\label{N-measure}
\mathcal{N}^{\prime}_{\text{HCLA}}=&\int_{\alpha_-}^1\gamma^{\prime}(p)\mathrm{d}p
=\left[\frac{1}{d}\ln |h(\alpha,p)|-\frac{2(d-1)\alpha}{d^2\sqrt{\Delta}}\ln\left|\frac{p-p_+}{p-p_-}\right|
\right]\bigg|_{\alpha_-}^{1}\nonumber\\
=&\frac{1}{d}\ln\left|\frac{h(\alpha,1)}{h(\alpha,\alpha_-)}\right|+
\frac{2(d-1)\alpha}{d^2\sqrt{\Delta}}\ln
\left|\frac{(1-p_-)(\alpha_{-}-p_+)}{(1-p_+)(\alpha_{-}-p_-)}\right|,
\end{align}
where $\Delta=\left[\left(d-\frac{2(d-1)}{d}\right)\alpha+d\right]^2-4(d-1)\alpha(d-1-\alpha)$ is the discriminant of the function $h(\alpha,p),$ and
$p_{\pm}=\frac{\left[d-\frac{2(d-1)}{d}\right]\alpha+d
\pm\sqrt{\Delta}}{2(d-1)\alpha}$ are the two real roots of $h(\alpha,p).$ In particular, Eq. \eqref{N-measure} reduces to the case of $d=2$ in Ref. \cite{Shr-2018}.

In the following, we give a simple example to show that the values of the normalized HCLA measure $\mathcal{N}^{\prime}_{\text{HCLA}}$ increase with the parameter $\alpha$.

{\it Example 2:} For $d=3$, Eq. \eqref{N-measure} reduces to
\begin{equation}
\mathcal{N}^{\prime}_{\text{HCLA}}=\frac{1}{3}\ln
\left|\frac{-\frac{2}{3}\alpha-1}
{2\alpha\alpha_{-}^2-(\frac{5}{3}\alpha+3)\alpha_{-}+2-\alpha}\right|+
\frac{4\alpha}{3\sqrt{97\alpha^2-54\alpha+81}}\ln
\left|\frac{(1-p_-)(\alpha_{-}-p_+)}{(1-p_+)(\alpha_{-}-p_-)}\right|,
\end{equation}
where $\alpha_-=\frac{3\left(\alpha+1-\sqrt{\alpha^2-\frac{2}{3}\alpha+1}\right)}
{4\alpha}$ , and $p_{\pm}=\frac{5\alpha+9\pm\sqrt{97\alpha^2-54\alpha+81}}{12\alpha}$.
A plot of the values of $\mathcal{N}^{\prime}_{\text{HCLA}}$ is given in Fig. \ref{fig4}. One can see that $\mathcal{N}^{\prime}_{\text{HCLA}}$ as a function of non-Markovian parameter $\alpha$, which increases monotonically with $\alpha$.

\begin{figure}[htbp]
\centering
\setlength{\abovecaptionskip}{0.05cm}
{\includegraphics[width=7cm,height=6cm]{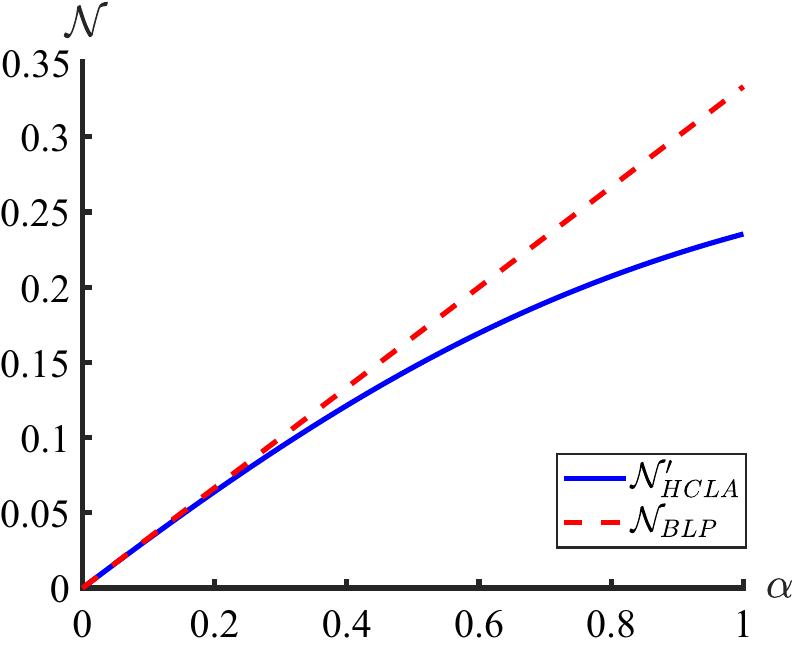}}
\caption{\footnotesize{A plot of the normalized HCLA measure $\mathcal{N}^{\prime}_{\text{HCLA}}$ (solid, blue curve) and the BLP measure $\mathcal{N}_{\text{BLP}}$ (dashed, red line) as a function of the non-Markovian parameter $\alpha$, respectively.}}\label{fig4}
\end{figure}

\section{The singularity of intermediate map is not pathological}
\label{sec:5}
As discussed above, $p=\alpha_{-}$ is the singular point of the intermediate map $\mathcal{E}^{\text{int}}(p)\; (p\in[p_*,p^*])$, i.e., the generalized Weyl map is temporarily noninvertible at the point $p=\alpha_{-}$, after which the invertibility is recovered. In this section, we find the following two phenomena:  first, the singularity is not pathological in the sense of the density operator; second, although
the generator in the master equation Eq. \eqref{master-equation} has a
 singularity, and the dynamics is still regular (i.e., its solution does not contain singularity) \cite{Chr2010}.

For any initial state $\rho(0)=[\rho_{ij}]_{d\times d}$, the intermediate map evolves the initial state $\rho(0)$ to the diagonal state $\rho(p)=\text{diag}(\rho_{00},\cdots,\rho_{d-1,d-1})$ at the point $p=\alpha_-.$  We illustrate this aspect with the following mathematical arguments. Note that the eigenvalues in Eq. \eqref{eigenvalue} satisfy $\sum_{j=0}^{d-1}\lambda_j=d$; for the intermediate map $\mathcal{E}^{\text{int}}(p)$ of Eq. \eqref{int-map} acts on the initial state $\rho(0)$, one can expand its to the matrix form as follows:
\begin{equation}\label{evolve-p-matrix}
\rho(p)=
\begin{pmatrix}
\rho_{00} & \frac{G(p^*)}{G(p_*)}\rho_{01} & \cdots & \frac{G(p^*)}{G(p_*)}\rho_{0,d-1}\\
\frac{G(p^*)}{G(p_*)}\rho_{10} & \rho_{11} & \cdots & \frac{G(p^*)}{G(p_*)}\rho_{1,d-1}\\
\vdots  &  \vdots & \ddots & \vdots \\
\frac{G(p^*)}{G(p_*)}\rho_{d-1,0} & \frac{G(p^*)}{G(p_*)}\rho_{d-1,1} & \cdots & \rho_{d-1,d-1}
\end{pmatrix},
\end{equation}
where $\frac{G(p^*)}{G(p_*)}=\frac{(p^*-\alpha_-)(p^*-\alpha_+)}
{(p_*-\alpha_-)(p_*-\alpha_+)}$. When $p=p^*=\alpha_-$ and $p_*<p^*$, Eq. \eqref{evolve-p-matrix} reduces to the diagonal matrix $\rho(\alpha_-)=\text{diag}(\rho_{00},\cdots,\rho_{d-1,d-1})$. In this sense, all the off-diagonal terms of $\rho(p)$ are vanishing at the singular point $p=\alpha_-,$ which corresponds to an instance of noninvertibility for generalized Weyl map. In addition, one can find that all the initial states $\rho(0)$ become indistinguishable at this point $p=\alpha_-$ momentarily, since the corresponding evolve states $\rho(\alpha_-)$ are of the same diagonal matrix. However, the singularity of the intermediate map is not pathological in the sense of the density operator since the evolve state $\rho(\alpha_-)$ is still a density operator.

On the other hand, if the intermediate map $\mathcal{E}^{\text{int}}(p)$ of Eq. \eqref{int-map} acts on the evolve diagonal state $\rho(\alpha_-)$, one obtains
$$
\mathcal{E}^{\text{int}}(p)[\rho(\alpha_-)]=
\frac{1}{d}\left[\lambda_0\rho(\alpha_-)
+\sum_{i=1}^{d-1}\lambda_iU_{di}\rho(\alpha_-) U^\dag_{di}\right]=
\frac{1}{d}\sum_{j=0}^{d-1}\lambda_j\rho(\alpha_-)=\rho(\alpha_-),
$$
where $\lambda_j, j=0,1,\ldots,d-1$ are defined by Eq. \eqref{eigenvalue}. This means that the action of the intermediate map
$\mathcal{E}^{\text{int}}(p)$ on the state $\rho(\alpha_-)$ remains invariant. More specifically, the expressions of the eigenvalues $\lambda_j$ contain the infinite term $\frac{(\alpha_{-}-p^*)(\alpha_{+}-p^*)}
{(\alpha_{-}-p_*)(\alpha_{+}-p_*)}$ (in the sense of $p_*=\alpha_{-}$ and $p_*<p^*$), which corresponding to the singularity of the intermediate map. The infinite term (i.e., the singularity) has no effect on the density operator $\rho(\alpha_-)$ as it multiply with the off-diagonal terms of the density operator $\rho(\alpha_-)$, which are vanishing.

Similarly, although the decoherence rates $\gamma_i(p)=\gamma(p) \; (i=1,\ldots,d-1)$ in the generator of Eq. \eqref{master-equation} are divergent at the singular point $p=\alpha_-,$ the singularity has no effect on the solution of Eq. \eqref{master-equation}. In fact, these terms $U_{di}\rho(\alpha_-)U^\dag_{di}-\rho(\alpha_-)=0$ for all $i=1,\ldots,d-1$, then they multiply with the decoherence rates $\gamma_i(p)$ are vanishing. Therefore, the solution of Eq. \eqref{master-equation} does not contain singularity, it is regular \cite{Chr2010}.

\section{Quantifying non-Markovianity by the trace distance measure}
\label{sec:6}
In Ref. \cite {BLP-2009}, Breuer-Laine-Piilo (BLP) proposed a non-Markovianity measure by the trace distance of a pair initial states. For a CPTP map $\mathcal{E}(p)$ acts on the initial states $\rho_i(0)$, i.e., $\rho_i(p)=\mathcal{E}(p)[\rho_i(0)], i=1,2$,
one can define the rate of change of the trace distance by
\begin{equation}
\sigma(p,\rho_i(0))=\frac{\mathrm{d}}{\mathrm{d}p}D[\rho_1(p),\rho_2(p)],
\end{equation}
where the trace distance $D$ is defined by Eq. \eqref{trace-distance}. If $\sigma(p,\rho_i(0))\leq0,$ then the dynamical map $\mathcal{E}(p)$ is Markovian. Otherwise, it is non-Markovian for some time-like parameter $p\in[0,1]$. The non-Markovianity measure of BLP is defined by
\begin{equation}\label{max-nonMark}
\mathcal{N}_{\text{BLP}}=\max_{\rho_i(0)}\int_{\sigma>0}dp\;\sigma(p,\rho_i(0)).
\end{equation}

For a time-like parameter $p\in[0,1]$, we consider the generalized Weyl map $\mathcal{E}(p)$ of Eq. \eqref{special-Weyl} acts on a pair of initial states $\rho_i(0), i=1,2.$ The qubit pure states of the trace distance measure was analyzed in Ref. \cite{Shr-2018}. In this section, we shall generalize to the high-dimensional cases. For the general pure state $|\psi\rangle,$ it can be parameterized as $|\psi\rangle=\sum_{k=0}^{d-1}\eta_k|k\rangle,$
where
\begin{equation}\label{general-pure-state}
\begin{cases}
\eta_0=\cos\theta_{d-1},\\
\eta_1=\sin\theta_{d-1}\cos\theta_{d-2}e^{\mathrm{i}\phi_{d-1}},\\
\cdots\\
\eta_{d-2}=\sin\theta_{d-1}\sin\theta_{d-2}\cdots
\sin\theta_{2}\cos\theta_{1}e^{\mathrm{i}\phi_{2}},\\
\eta_{d-1}=\sin\theta_{d-1}\sin\theta_{d-2}\cdots
\sin\theta_{2}\sin\theta_{1}e^{\mathrm{i}\phi_{1}},
\end{cases}
\end{equation}
and $\theta_k\in[0,\frac{\pi}{2}], \phi_{k}\in[0,2\pi)$ for $k=1,\ldots,d-1.$ However, the calculation of trace distance becomes quite complicated under this parameterization, we turn to a simpler approach to study the non-Markovianity measure in the sense of BLP. That is, we restrict the pair of orthogonal initial states to be mutually unbiased bases (MUBs). Two sets of the orthonormal bases $\{|\psi_i\rangle\}_{i=0}^{d-1}$ and $\{|\phi_j\rangle\}_{j=0}^{d-1}$ in $\mathbb{C}^d$ are called mutually unbiased if $\left|\langle\psi_i|\phi_j\rangle\right|^2=\frac{1}{d}$ for all $i,j=0,\ldots,d-1.$ It is well known that the number of MUBs in $\mathbb{C}^d$ is at most $d+1$ \cite{Woot-1989}. If $d$ is a prime power, then there exist $d+1$ MUBs, i.e., a complete set of MUBs, but for other dimensions the maximal number of MUBs remains unknown \cite{Dur-2004}.

{\sf Remark 4:} As illustrated in Ref. \cite{BLP-2009}, the maximization over the pair of initial states $\rho_i(0)$ in Eq. \eqref{max-nonMark} can be performed by drawing a sufficiently large sample of random pairs of initial states. The numerical results show that the non-Markovianity measure $\mathcal{N}_{\text{BLP}}$ reaches maximum value when the pair of initial states is orthogonal.
Hence, this finding is the reason why we choose orthogonal MUBs as the pair of initial states to study the non-Markovianity measure in terms of BLP.

In the following, we first consider the simple case of $d=3$, and then extend our analysis to arbitrarily high-dimensional cases.

For $d=3$, we choose two orthogonal initial states $|\psi_i\rangle$ from the same set of MUBs $\mathcal{B}_k \;(k=0,1,2,3)$, which are given by \cite{Wie-2011}
\begin{align}\label{MUBs}
&\mathcal{B}_0=\left\{(1,0,0)^T, (0,1,0)^T,(0,0,1)^T\right\},\nonumber\\
&\mathcal{B}_1=\left\{\frac{1}{\sqrt{3}}(1,1,1)^T,
\frac{1}{\sqrt{3}}(1,\omega_3,\omega_3^2)^T,
\frac{1}{\sqrt{3}}(1,\omega_3^2,\omega_3)^T\right\}, \nonumber\\
&\mathcal{B}_2=\left\{\frac{1}{\sqrt{3}}(1,\omega_3,\omega_3)^T,
\frac{1}{\sqrt{3}}(1,\omega_3^2,1)^T,
\frac{1}{\sqrt{3}}(1,1,\omega_3^2)^T\right\}, \nonumber\\
&\mathcal{B}_3=\left\{\frac{1}{\sqrt{3}}(1,\omega_3^2,\omega_3^2)^T,
\frac{1}{\sqrt{3}}(1,1,\omega_3)^T,
\frac{1}{\sqrt{3}}(1,\omega_3,1)^T\right\}.
\end{align}
Here $\omega_3=e^{2\pi\mathrm{i}/3},$ and $T$ denotes the transposition of vectors.

The evolve states $\rho_{i}(p)=\mathcal{E}(p)[\rho_{i}(0)]$ are given by
\begin{equation}\label{d=3-Weyl-c}
\rho_{i}(p)=\left[1-\frac{2}{3}\kappa(p)\right]\rho_{i}(0)+\frac{\kappa(p)}{3}
\left(U_{3}\rho_{i}(0)U_{3}^\dag+U_{6}\rho_{i}(0)U_{6}^\dag\right),
\end{equation}
where $\kappa(p)=p\left[1+\alpha(1-\frac{2}{3}p)\right]$, and $\rho_{i}(0)=|\psi_i\rangle\langle\psi_i|$ for $i=1,2$. The trace distance between $\rho_{1}(p)$ and $\rho_{2}(p)$ is given by
\begin{equation}\label{p-trace-dis}
D[\rho_{1}(p),\rho_{2}(p)]=\frac{1}{2}\tr\sqrt{[\rho_{1}(p)-\rho_{2}(p)]^2}.
\end{equation}

In Appendix D, we show that the trace distance $D[\rho_{1}(p),\rho_{2}(p)]$ is given by
$$
D[\rho_{1}(p),\rho_{2}(p)]=
\begin{cases}
1,& \; \text{if}~ |\psi_i\rangle\in\mathcal{B}_0,\\
\frac{2\alpha}{3}\left|(p-\alpha_-)(p-\alpha_+)\right|,& \; \text{if}~ |\psi_i\rangle\in\mathcal{B}_k\setminus\mathcal{B}_0 ,
\end{cases}
$$
where $\alpha_{\pm}=\frac{3\left(\alpha+1\pm\sqrt{\alpha^2-\frac{2}{3}\alpha+1}\right)}{4\alpha}$.
In particular, when $\alpha=0$, we have $\kappa(p)=p$ and the trace distance reduces to
$$
D[\rho_{1}(p),\rho_{2}(p)]=
\begin{cases}
1,& \; \text{if}~ |\psi_i\rangle\in\mathcal{B}_0,\\
1-p,& \; \text{if}~ |\psi_i\rangle\in\mathcal{B}_k\setminus\mathcal{B}_0.
\end{cases}
$$

For the initial states $|\psi_i\rangle\in\mathcal{B}_0$, we find that the trace distance is a constant independent of $p$ and $\alpha$. However, our main purpose is to study the case where the trace distance $D$ is related to $p$ and $\alpha$. Therefore, we select the nine different pairs of initial states from $\mathcal{B}_k\setminus\mathcal{B}_0$ to investigate the trace distance measure.

For the pair of orthogonal initial states $|\psi_i\rangle\in\mathcal{B}_k\setminus\mathcal{B}_0$, the trace distance $D$
is depicted in Fig. \ref{fig5} for the non-Markovian parameter $\alpha=0,\;0.4,\;0.7,$ respectively. It can be observed that the generalized Weyl map $\mathcal{E}(p)$
exhibits non-Markovian behavior in the range $p\in(\alpha_-,1]$, as the trace distance increases monotonically within this interval. Moreover, the corresponding non-Markovian region expands as the non-Markovian parameter $\alpha$ increases.

\begin{figure}[htbp]
\centering
\setlength{\abovecaptionskip}{0.05cm}
{\includegraphics[width=7cm,height=6cm]{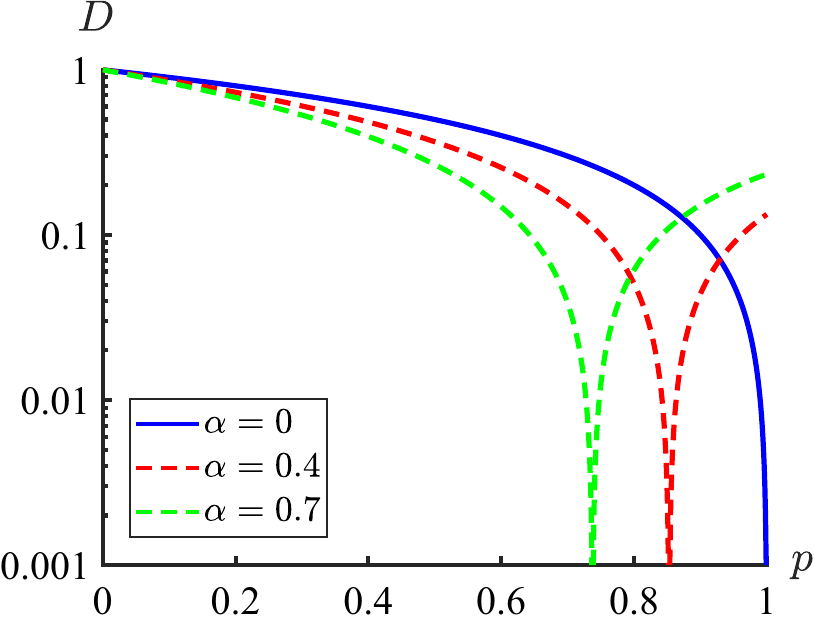}}
\caption{\footnotesize{A logarithmic plot (we just take logarithm to the vertical axis) of trace distance $D$ as a function of $p$. Three cases are presented: the non-Markovian parameter $\alpha=0$ (solid, blue curve), $\alpha=0.4$ (dashed, red curve), and $\alpha=0.7$ (dashed, green curve).}}\label{fig5}
\end{figure}

By taking over these nine different pairs of initial pure states $|\psi_i\rangle\in\mathcal{B}_k\setminus\mathcal{B}_0$, we can obtain the quantity of the non-Markovianity measure $\mathcal{N}_{\text{BLP}}$ as follows:
\begin{align}
\mathcal{N}_{\text{BLP}}=&\max_{|\psi_i\rangle\in\mathcal{B}_k\setminus\mathcal{B}_0}
\int_{\alpha_-}^1 \frac{\mathrm{d}D}{\mathrm{d}p}\mathrm{d}p=
\frac{2\alpha}{3}\int_{\alpha_-}^1 \left|2p-(\alpha_{+}+\alpha_-)\right|\mathrm{d}p  \nonumber\\
=&\frac{2\alpha}{3}\left|1-(\alpha_{+}+\alpha_-)+\alpha_{+}\alpha_-\right|
=\frac{\alpha}{3},
\end{align}
where the fourth equality holds because $\alpha_{+}+\alpha_-=\frac{3(\alpha+1)}{2\alpha}$ and $\alpha_{+}\alpha_{-}=\frac{3}{2\alpha}$. The BLP measure $\mathcal{N}_{\text{BLP}}$ (dashed, red line) is depicted in Fig. \ref{fig4}. One can see that in some intervals of the non-Markovian parameter $\alpha$, the quantification of non-Markovianity is in agreement between the BLP measure $\mathcal{N}_{\text{BLP}}$ and the normalized HCLA measure $\mathcal{N}^{\prime}_{\text{HCLA}}$.

Although the existence of a complete set of MUBs is unknown when the dimension of the system is not a prime power, one can always utilize all the known sets of MUBs to generalize the non-Markovianity measure in terms of BLP to high-dimensional cases. For the generalized Weyl map $\mathcal{E}(p)$ defined by Eq. \eqref{special-Weyl} acts on a pair of initial states $\rho_i(0)$ (which are the form of orthogonal MUBs), in Appendix D, we demonstrate the trace distance
$$
D[\rho_{1}(p),\rho_{2}(p)]=
\begin{cases}
1,& \; \text{if}~ |\psi_i\rangle\in\mathcal{B}_0,\\
\frac{(d-1)\alpha}{d}\left|(p-\alpha_-)(p-\alpha_+)\right|,& \; \text{if}~ |\psi_i\rangle\in\mathcal{B}_k\setminus\mathcal{B}_0.
\end{cases}
$$
Here $\alpha_{\pm}$ are given by Eq. \eqref{two-roots}, and the set $\mathcal{B}_0=\left\{e_j\right\}_{j=0}^{d-1}$, where the $(j+1)$-th component of the column vector $e_j$ is 1 and the rest are $0$. By taking over all known orthogonal MUBs as a pair of initial states, we can obtain the quantity of non-Markovianity measure in terms of BLP as follows:
\begin{equation}
\mathcal{N}_{\text{BLP}}=\max_{|\psi_i\rangle\in\mathcal{B}_k\setminus\mathcal{B}_0}
\int_{\alpha_-}^1 \frac{\mathrm{d}D}{\mathrm{d}p}\mathrm{d}p=
\frac{(d-1)\alpha}{d}\int_{\alpha_-}^1 \left|2p-(\alpha_{+}+\alpha_-)\right|\mathrm{d}p=\frac{\alpha}{d}.
\end{equation}
In particular, when $d=2$, the quantity of non-Markovianity measure in terms of BLP is the same as Ref. \cite{Shr-2018}.

\section{Conclusions and discussions}\label{sec:7}
In this paper, we have extended the results of Ref. \cite{Shr-2018} to the case of qudit by using a special class of generalized Weyl channel. They are characterized by the sum of Kraus operators where only part of Kraus operators are proportional
to the Weyl diagonal matrices and the rest are vanishing. We have explored the non-Markovianity by investigating the eigenvalues of the Choi matrix corresponding to the intermediate map. Specifically, we have provided a straightforward case for $d=3$, one can find that a singular point of the intermediate map occurring at the crossover of its $d$ eigenvalues.

Furthermore, we have demonstrated that the singularity in the intermediate map is not pathological in terms of density operators, and the solution of intermediate map remain regular. Additionally, we have quantified non-Markovianity by using the HCLA measure and the BLP measure, respectively. In particular, we have chosen the MUBs as the pair of orthogonal initial states to quantify the non-Markovianity based on the BLP measure.

Recently, a full characterization of quantum memory witness in the case of qubit system has been provided in Ref. \cite{Ban-2023}.
It is intriguing to extend quantum memory witness to high-dimensional scenarios by using generalized Weyl channels. Exploring the geometrical characterization of non-Markovianity \cite{Lor-2013,Abu-2024} under the generalized Weyl channels also offers a promising avenue for future research. Moreover, instead of restricting to a pair of orthogonal MUBs as initial states, we may select more general parameterized initial pure states, such as Eq. \eqref{general-pure-state}, which offers a challenge to study the non-Markovianity in terms of the BLP measure. These future research directions hold the potential to deepen our understanding of quantum dynamics and pave the way for practical applications in quantum information processing.

\section*{Data availability statement}
All data that support the findings of this study are included within the article (and any supplementary files).

\section*{Acknowledgments}
 This work is supported by Key Lab of Guangzhou for Quantum Precision Measurement under Grant No.202201000010, the Guangdong Basic and Applied Basic Research Foundation under Grant No.2023A1515012074 and the National  Natural  Science  Foundation of China (NSFC) under Grants No.12371458, No.12175147; Zhejiang Provincial Natural Science Foundation of China under Grant No.LZ24A050005; Jiangxi Provincial Natural Science Foundation of China, under Grants No. 20213BCJL22054.

\appendix

\section*{Appendix A. The Weyl matrices for the case of $d=3$}

The following are the Weyl matrices for $d=3$: $U_0:=U_{00}=\mathbbm{1}_3$ and
\begin{equation*}
U_1:=U_{01}=
\begin{pmatrix}
0& 1 &0\\
0& 0 &1\\
1& 0 &0
\end{pmatrix},
U_2:=U_{02}=
\begin{pmatrix}
0& 0 &1\\
1& 0 &0\\
0& 1 &0
\end{pmatrix},
U_3:=U_{10}=
\begin{pmatrix}
1& 0 &0\\
0& \omega_3 &0\\
0& 0 &\omega_3^2
\end{pmatrix},
U_4:=U_{11}=
\begin{pmatrix}
0& 1 &0\\
0& 0 &\omega_3\\
\omega_3^2& 0 &0
\end{pmatrix},
\end{equation*}

\begin{equation*}
U_5:=U_{12}=
\begin{pmatrix}
0& 0 &1\\
\omega_3& 0 &0\\
0& \omega_3^2 &0
\end{pmatrix},
U_6:=U_{20}=
\begin{pmatrix}
1& 0 &0\\
0& \omega_3^2 &0\\
0& 0 &\omega_3
\end{pmatrix},
U_7:=U_{21}=
\begin{pmatrix}
0& 1 &0\\
0& 0 &\omega_3^2\\
\omega_3& 0 &0
\end{pmatrix},
U_8:=U_{22}=
\begin{pmatrix}
0& 0 &1\\
\omega_3^2& 0 &0\\
0& \omega_3& 0
\end{pmatrix},
\end{equation*}
where $\omega_3=e^{\frac{2\pi\mathrm{i}}{3}}.$
\section*{Appendix B. The proof of Eqs. \eqref{Choi-representation} and \eqref{superoperator-representation}} 
We can use the following two identities
\begin{align}
&|Z\rangle\rangle=(Z\otimes\mathbbm{1}_d)|\mathbbm{1}_d\rangle\rangle, \label{vec-1}\\
&|XZY\rangle\rangle=(X\otimes Y^T)|Z\rangle\rangle, \label{vec-2}
\end{align}
to prove Eqs. \eqref{Choi-representation} and \eqref{superoperator-representation}, where the operators $X,Y,Z\in\mathbb{C}^{d\times d},$ and $T$ is the transpose of matrix.

The Choi matrix representation is given by
\begin{align*}
C=&\sum_{i,j=0}^{d-1}\mathcal{E}(|i\rangle\langle j|)\otimes |i\rangle\langle j|=\sum_{a=0}^{d^2-1}\sum_{i,j=0}^{d-1}K_a|i\rangle\langle j|K_a^\dag\otimes |i\rangle\langle j| \nonumber\\
=&\sum_{a=0}^{d^2-1}(K_a\otimes \mathbbm{1}_d)\left(\sum_{i,j=0}^{d-1}|i\rangle\langle j|\otimes|i\rangle\langle j|\right)(K_a^\dag\otimes \mathbbm{1}_d) \nonumber\\
=&\sum_{a=0}^{d^2-1}(K_a\otimes \mathbbm{1}_d)|\mathbbm{1}_d\rangle\rangle\langle\langle\mathbbm{1}_d|(K_a^\dag\otimes \mathbbm{1}_d)
=\sum_{a=0}^{d^2-1}|K_a\rangle\rangle\langle\langle K_a|,
\end{align*}
where $|\mathbbm{1}_d\rangle\rangle=\sum_{i=0}^{d-1}|ii\rangle,$ and the last equality holds as Eq. \eqref{vec-1}.

The superoperator representation is given by
\begin{equation*}
\widehat{\mathcal{E}}|X\rangle\rangle=|\mathcal{E}(X)\rangle\rangle
=\sum_{a=0}^{d^2-1}|K_a X K_a^\dag\rangle\rangle=\sum_{a=0}^{d^2-1}(K_a\otimes \overline{K}_a)|X\rangle\rangle,
\end{equation*}
where the last equality holds as Eq. \eqref{vec-2}.

\section*{Appendix C. The derivation of Choi matrix for the intermediate map $\mathcal{E}(p^*,p_*)$} 
By the superoperator representation Eq. \eqref{superoperator-representation} and the Kraus operators Eq. \eqref{parameterization-kraus} of generalized Weyl channel $\mathcal{E}(p_*)$, one has
\begin{equation}
\widehat{\mathcal{E}}(p_*)=\sum_{i=0}^{d-1}K_{di}\otimes\overline{K}_{di}
=\text{diag}\left(D_0,D_1,\ldots,D_{d-1}\right).
\end{equation}
Here $D_i\;(i=0,1,\ldots,d-1)$ are $d\times d$ diagonal matrices whose $(i+1)$-th diagonal element being $1$, the other diagonal elements are $1-\frac{d-1}{d}\kappa(p_*)+\frac{\kappa(p_*)}{d}
\sum_{i=1}^{d-1}\omega_d^i=1-\kappa(p_*)$ due to the $d$th root of unity $\omega_d=e^{\frac{2\pi\mathrm{i}}{d}}$ satisfies the property $\sum_{i=0}^{d-1}\omega_d^i=0$. Similarly, we can also obtain the matrix form of the superoperator $\widehat{\mathcal{E}}(p^*)$ by replacing $p_*$ with $p^*$.

Combined Eq. \eqref{superoperator} with Eq. \eqref{invertible}, the superoperator of the intermediate map $\mathcal{E}(p^*,p_*)$ is given by
\begin{equation}\label{inter-super}
\widehat{\mathcal{E}}(p^*,p_*)|X\rangle\rangle=
|\mathcal{E}(p^*,p_*)(X)\rangle\rangle=
|\mathcal{E}(p^*)[\mathcal{E}^{-1}(p_*)(X)]\rangle\rangle=
\widehat{\mathcal{E}}(p^*)|\mathcal{E}^{-1}(p_*)(X)\rangle\rangle=
\widehat{\mathcal{E}}(p^*)\widehat{\mathcal{E}}^{-1}(p_*)|X\rangle\rangle.
\end{equation}
Hence, by the relations of Eqs. \eqref{reshuffling} and \eqref{inter-super}, the Choi matrix for the intermediate map $\mathcal{E}(p^*,p_*)$ is given by
\begin{equation*}
\chi(\alpha,p^*,p_*)=\left[\widehat{\mathcal{E}}(p^*,p_*)\right]^{\mathcal{R}}=\left[\widehat{\mathcal{E}}(p^*)
\widehat{\mathcal{E}}^{-1}(p_*)\right]^{\mathcal{R}},
\end{equation*}
which is the form of \eqref{Choi-inter}. In particular, the Choi matrix of Eq. \eqref{Choi-inter} reduces to
$$
\begin{pmatrix}
1& \;0& \;0& \; \frac{1-\kappa(p^*)}{1-\kappa(p_*)} \\
0& \;0&\; 0&\; 0\\
0& \;0& \;0& \;0\\
\frac{1-\kappa(p^*)}{1-\kappa(p_*)}&\; 0&\; 0&\; 1
\end{pmatrix},
$$
which is consistent with the case of $d=2$ in Ref. \cite{Shr-2018}.
In addition, for the case of $d=3$, the Choi matrix of Eq. \eqref{Choi-inter} is given  by
$$
\begin{pmatrix}
1& \;0& \;0& \;0& \;\frac{1-\kappa(p^*)}{1-\kappa(p_*)}& \;0& \;0&\; 0& \; \frac{1-\kappa(p^*)}{1-\kappa(p_*)}\\
0& \;0&\; 0&\; 0& \;0&\; 0&\; 0& \;0&\; 0\\
0& \;0& \;0& \;0& \;0& \;0& \;0& \;0& \;0\\
0& \;0& \;0& \;0& \;0& \;0& \;0& \;0& \;0\\
\frac{1-\kappa(p^*)}{1-\kappa(p_*)}&\; 0& \;0&\; 0& \;1&\; 0&\; 0&\; 0& \; \frac{1-\kappa(p^*)}{1-\kappa(p_*)}\\
0& \;0& \;0& \;0& \;0& \;0& \;0& \;0& \;0\\
0& \;0& \;0& \;0& \;0& \;0& \;0& \;0& \;0\\
0& \;0& \;0& \;0& \;0& \;0& \;0& \;0& \;0\\
\frac{1-\kappa(p^*)}{1-\kappa(p_*)}&\; 0& \;0&\; 0& \; \frac{1-\kappa(p^*)}{1-\kappa(p_*)}& \;0& \;0& \;0&\; 1
\end{pmatrix}.
$$

\section*{Appendix D. The derivation of trace distance for the pair of orthogonal MUBs} 
\textbf{The case of \boldmath{$d=3$}:} Combined Eq. \eqref{d=3-Weyl-c} with Eq. \eqref{p-trace-dis}, one can obtain the trace distance $D[\rho_{1}(p),\rho_{2}(p)]$. In the following, for $d=3, \kappa(p)=p\left[1+\alpha(1-\frac{2}{3}p)\right],$ $\alpha_{\pm}=\frac{3\left(\alpha+1\pm\sqrt{\alpha^2-\frac{2}{3}\alpha+1}\right)}
{4\alpha}$, we divide the selected pair of orthogonal MUBs into two cases for discussions.
\begin{enumerate}
\item [\rm (a)] If $|\psi_i\rangle\in\mathcal{B}_0,$ for instance, we choose $|\psi_1\rangle=(1,0,0)^T, |\psi_2\rangle=(0,1,0)^T.$ One has
\begin{equation}\label{rho-difference-1}
\rho_{1}(p)-\rho_{2}(p)=
\begin{pmatrix}
1 &\; 0 &\; 0 \\
0 &\; -1 &\; 0 \\
0 &\; 0 &\; 0
\end{pmatrix}.
\end{equation}
The corresponding singular values are $1, 1$, and $0$. Substituting Eq. \eqref{rho-difference-1} into Eq. \eqref{p-trace-dis}, one can get
\begin{equation}\label{trace-dis-B-0}
D[\rho_{1}(p),\rho_{2}(p)]=1.
\end{equation}
Similarly, the values of trace distance are the same as Eq. \eqref{trace-dis-B-0} for the other two choices of initial states. In particular, when $\alpha=0$, we have $\kappa(p)=p$ and the trace distance $D[\rho_{1}(p)-\rho_{2}(p)]=1.$

\item [\rm (b)] If $|\psi_i\rangle\in\mathcal{B}_k \;(k=1,2,3),$
for instance, we choose one pairs of initial pure states $|\psi_1\rangle=\frac{1}{\sqrt{3}}(1,1,1)^T$ and $|\psi_2\rangle=\frac{1}{\sqrt{3}}(1,\omega_3,\omega_3^2)^T$. Denote the difference between the evolve states $\rho_1(p)$ and $\rho_2(p)$ as  $A=\rho_{1}(p)-\rho_{2}(p)$, after some algebraic calculations, one obtains
\begin{equation}\label{rho-difference}
A=\frac{1-\kappa(p)}{3}
\begin{pmatrix}
0 &\; 1-\omega_3^2 &\; 1-\omega_3 \\
1-\omega_3 &\; 0 &\; 1-\omega_3^2 \\
1-\omega_3^2 &\; 1-\omega_3 &\; 0
\end{pmatrix}.
\end{equation}
One can easily obtain that the singular values of Eq. \eqref{rho-difference} are $[1-\kappa(p)]^2, [1-\kappa(p)]^2$ and $ 0.$ Substituting Eq. \eqref{rho-difference} into Eq. \eqref{p-trace-dis}, one can get
\begin{equation}\label{d=3-trace-distance}
D[\rho_{1}(p),\rho_{2}(p)]=\left|1-\kappa(p)\right|=
\frac{2\alpha}{3}\left|(p-\alpha_-)(p-\alpha_+)\right|.
\end{equation}
In particular, when $\alpha=0$, we have $\kappa(p)=p$ and the trace distance $D[\rho_{1}(p)-\rho_{2}(p)]=1-p.$

Similarly, the values of trace distance are the same as Eq. \eqref{d=3-trace-distance} for the other eight choices of initial states. Indeed, denote the difference between the other pairs of evolve states $\rho_1(p)$ and $\rho_2(p)$ as $A^{\prime}$, after some algebraic calculations, one finds that there always exists a complex invertible matrix $P\in\mathbb{C}^{3\times3}$, such that $P^{-1}A'^2P=A^2,$ i.e., $A^2$ is similar to $A'^2$. Therefore, $A'^2$ and $A^2$ have the same eigenvalues \cite{Hor-2013}. Then the values of trace distance are coincided with Eq. \eqref{d=3-trace-distance} for the other eight choices of initial states.
\end{enumerate}

Although the existence of a complete set of MUBs is unknown when the dimension of the system is not a prime power, one can always utilize all the known sets of MUBs to generalize the non-Markovianity measure in terms of BLP to high-dimensional cases.

Before showing the expression of trace distance $D\left[\rho_1(p),\rho_2(p)\right]$ for high-dimensional cases, we need the following technical lemma.

\begin{lemma}\label{lem:1}
Let $d\geq2$ be an integer and $\omega_d=e^{2\pi\mathrm{i}/d}$ be a
$d$th primitive root of unit. Then we have the following identity
\begin{equation}\label{poly-coefficients}
\sum_{k=0}^{d-1}\sum_{\substack{i,j=0\\i+j\equiv k \;\text{mod}\; d}}^{d-1}\left(1-\omega_d^{d-i}\right)
\left(1-\omega_d^{d-j}\right)x^k=
\begin{cases}
d^2, &~ \text{if} ~x=1,\omega_d,\\
0, &~ \text{if} ~ x=\omega_d^2,\cdots,\omega_d^{d-1}.
\end{cases}
\end{equation}
\end{lemma}

\begin{proof}
Fixing $k=0$ on the left hand of Eq. \eqref{poly-coefficients}, one has
\begin{equation}\label{the-first-term}
\sum_{\substack{i,j=0\\i+j\equiv 0 \;\text{mod}\; d}}^{d-1}\left(1-\omega_d^{d-i}\right)
\left(1-\omega_d^{d-j}\right)=\sum_{\substack{i,j=0\\i+j\equiv 0 \;\text{mod}\; d}}^{d-1}\left[2-\left(\omega_d^{d-i}+\omega_d^{d-j}\right)\right]=2d,
\end{equation}
where the last equality holds because the primitive root of unit satisfies $\sum_{m=0}^{d-1}\omega_d^m=0.$
For each fixed $k\in[1,d-1],$ set $x=\omega_d^s\;(s=0,1,\ldots,d-1)$ on the left hand of Eq. \eqref{poly-coefficients}, one has
\begin{equation}\label{the-latter-term}
\sum_{\substack{i,j=0\\i+j\equiv k \;\text{mod}\; d}}^{d-1}\left[1+\omega_d^{-k}-\left(\omega_d^{d-i}+\omega_d^{d-j}\right)
\right]\omega_d^{sk}=
d\left(1+\omega_d^{-k}\right)\omega_d^{sk}.
\end{equation}
Summing these terms of Eq. \eqref{the-latter-term} for $k\in[1,d-1]$, we have
\begin{equation}\label{sum-latter-term}
\sum_{k=1}^{d-1}\sum_{\substack{i,j=0\\i+j\equiv k \;\text{mod}\; d}}^{d-1}\left[1+\omega_d^{-k}-\left(\omega_d^{d-i}+\omega_d^{d-j}\right)
\right]\omega_d^{sk}=
d\sum_{k=1}^{d-1}\left[\omega_d^{sk}+\omega_d^{(s-1)k}\right]
=
\begin{cases}
d(d-2), ~\text{if} ~s=0,1,\\
-2d, ~\text{if} ~s=2,\cdots,d-1.
\end{cases}
\end{equation}
\end{proof}
Combined Eq. \eqref{the-first-term} with Eq. \eqref{sum-latter-term}, one obtains Eq. \eqref{poly-coefficients}.  \qed

\textbf{The high-dimensional cases:} The generalized Weyl map $\mathcal{E}(p)$ acts on the initial states $\rho_i(0)=|\psi_i\rangle\langle\psi_i|$, i.e., $\rho_i(p)=\mathcal{E}(p)[\rho_i(0)],$ which satisfies
\begin{equation}\label{d-Weyl-c}
\rho_{i}(p)=\left[1-\frac{d-1}{d}\kappa(p)\right]\rho_{i}(0)+
\frac{\kappa(p)}{d}\sum_{i=1}^{d-1}U_{di}\rho_{i}(0)U_{di}^\dag.
\end{equation}
where $\kappa(p)=p\left[1+\alpha(1-\frac{d-1}{d}p)\right],$ and $\rho_{i}(0)=|\psi_i\rangle\langle\psi_i|$ for $i=1,2$. Combined Eq. \eqref{d-Weyl-c} with Eq. \eqref{p-trace-dis}, one can obtain the trace distance $D[\rho_{1}(p),\rho_{2}(p)]$. In the following, for the high-dimensional cases, we divide the selected pair of orthogonal MUBs into two cases for discussions.
\begin{enumerate}
\item [\rm (a)] If $|\psi_i\rangle\in\mathcal{B}_0=\left\{e_j\right\}_{j=0}^{d-1}$, similar to the case of $d=3$, one obtains the trace distance $D[\rho_{1}(p),\rho_{2}(p)]=1$, where the $(j+1)$-th component of the column vector $e_j$ is 1 and the rest are $0$.

\item [\rm (b)] If $|\psi_i\rangle\in\mathcal{B}_k\setminus\mathcal{B}_0$, for instance, we choose the following two $d$-dimensional orthogonal column vectors
$$
|\psi_1\rangle=\frac{1}{\sqrt{d}}(1,1,\cdots,1,1)^T,~ |\psi_2\rangle=\frac{1}{\sqrt{d}}(1,\omega_d,\cdots,\omega_d^{d-2},
\omega_d^{d-1})^T,~\omega_d=e^{\frac{2\pi\mathrm{i}}{d}},
$$
as a pair of initial pure states. Denote the difference between the evolve states $\rho_1(p)$ and $\rho_2(p)$ as  $A=\rho_{1}(p)-\rho_{2}(p)$, after some algebraic calculations, one can obtain the Hermitian matrix
\begin{equation}\label{high-rho-difference}
A=\frac{1-\kappa(p)}{d}
\begin{pmatrix}
0 &\; 1-\omega_d^{d-1} &\; \cdots&\; 1-\omega_d^2 &\;1-\omega_d \\
1-\omega_d &\; 0 &\; \cdots&\; 1-\omega_d^3 &\;1-\omega_d^2 \\
\vdots &\; \vdots &\; \ddots &\; \vdots &\; \vdots \\
1-\omega_d^{d-2} &\; 1-\omega_d^{d-3} &\; \cdots&\; 0 &\;1-\omega_d^{d-1}\\
1-\omega_d^{d-1} &\; 1-\omega_d^{d-2} &\; \cdots&\; 1-\omega_d&\; 0
\end{pmatrix}.
\end{equation}
The form of $A$ is actually a circulant matrix \cite{Hor-2013}, which can be expressed by
\begin{equation}
A=\frac{1-\kappa(p)}{d}\sum_{i=0}^{d-1}\left(1-\omega_d^{d-i}\right)J^i,
\end{equation}
where $J^0=\mathbbm{1}_d$ and
$$
J^i=
\begin{pmatrix}
0 &\;\mathbbm{1}_{d-i}\\
\mathbbm{1}_i &\;0
\end{pmatrix}
$$
are the basic circulant matrices for all $1\leq i\leq d-1$. One can obtain the singular values of Eq. \eqref{high-rho-difference} are $[1-k(p)]^2$ ($2$ multiplicity) and $0$ ($d-2$ multiplicity). Indeed, by the property of circulant matrix, we know that the product of two circulant matrices is still a circulant matrix. That is,
\begin{align}
A^2=&\frac{[1-\kappa(p)]^2}{d^2}
\left[\sum_{i=0}^{d-1}\left(1-\omega_d^{d-i}\right)J^i\right]
\left[\sum_{j=0}^{d-1}\left(1-\omega_d^{d-j}\right)J^j\right] \nonumber\\
=&\frac{[1-\kappa(p)]^2}{d^2}\sum_{i,j=0}^{d-1}\left(1-\omega_d^{d-i}\right)
\left(1-\omega_d^{d-j}\right)J^{(i+j)\; \text{mod} \; d}\nonumber\\
=&\frac{[1-\kappa(p)]^2}{d^2}\sum_{k=0}^{d-1}\sum_{\substack{i,j=0\\i+j\equiv k \;\text{mod}\; d}}^{d-1}\left(1-\omega_d^{d-i}\right)\left(1-\omega_d^{d-j}\right)J^k
\end{align}
is a circulant matrix. Define
$$
f(x)=\sum_{k=0}^{d-1}a_k x^k,
$$
which is a polynomial with degree no larger than $d-1$. Here the coefficients of the polynomial $f(x)$ are given by $$a_k=\frac{[1-\kappa(p)]^2}{d^2}\sum_{\substack{i,j=0\\i+j\equiv k \;\text{mod}\; d}}^{d-1}\left(1-\omega_d^{d-i}\right)\left(1-\omega_d^{d-j}\right)$$ for $0\leq k\leq d-1.$ Hence, we have $A^2=f(J).$ By the property of circulant matrix $A^2$, we know that the eigenvalues of $A^2$ are $f(1), f(\omega_d), f(\omega_d^2), \cdots, f(\omega_d^{d-1})$ \cite{Hor-2013}.
Moreover, by Lemma \ref{lem:1}, one has
$$
f(\omega_d^k)=
\begin{cases}
[1-\kappa(p)]^2,& ~ \text{if}~ k=0,1,\\
0, & ~ \text{if}~ k=2,\cdots,d-1,
\end{cases}
$$

Therefore, we can get the trace distance
\begin{equation}\label{high-trace-distance}
D[\rho_{1}(p),\rho_{2}(p)]=\left|1-\kappa(p)\right|=
\frac{(d-1)\alpha}{d}\left|(p-\alpha_-)(p-\alpha_+)\right|.
\end{equation}
where $\alpha_{\pm}$ are given by Eq. \eqref{two-roots}. In particular, when $\alpha=0$, we have $\kappa(p)=p$ and the trace distance $D[\rho_{1}(p)-\rho_{2}(p)]=1-p.$

For the other known orthogonal initial states, denote the difference between the other pairs of evolve states $\rho_1(p)$ and $\rho_2(p)$ as $A^{\prime}$. There always exists a complex invertible matrix $P\in\mathbb{C}^{d\times d}$, such that $P^{-1}A'^2P=A^2,$ i.e., $A^2$ is similar to $A'^2$. Therefore, $A'^2$ and $A^2$ have the same eigenvalues \cite{Hor-2013}. Then the values of the trace distance $D$ are coincided with Eq. \eqref{high-trace-distance}.
\end{enumerate}

\end{document}